\documentclass[reqno, a4paper]{amsart}
\usepackage{amssymb, amscd, amsfonts,  amsthm}
\usepackage[mathscr]{eucal}
\usepackage{graphicx}
\usepackage[all,poly]{xy} 
 \usepackage[all]{xy} 

\newtheorem{theorem}{Theorem}[section]

\newtheorem{corollary}[theorem]{Corollary}
\newtheorem{proposition}[theorem]{Proposition}

\theoremstyle{definition}

\newcommand{\remove}[1]{}
\newcommand{\luk}{\L u\-ka\-s\-ie\-wicz}

\DeclareMathOperator{\Boole}{\mathsf{B}}
\DeclareMathOperator{\interval}{[0,1]}

\DeclareMathOperator{\Spec}{\mathsf{Spec}}
 
 \DeclareMathOperator{\dist}{\mathsf{dist}}

\DeclareMathOperator{\prim}{\mathsf{prim}}

\DeclareMathOperator{\McNm}{\mathcal M([0,1]^{\it m})}
\DeclareMathOperator{\McNn}{\mathcal M([0,1]^{\it n})}

 \title[Polytime reductions of AF-algebraic problems]
{Polytime reductions of  AF-algebraic problems}

\author{Daniele Mundici}
 
\address[D. Mundici]{Department of
Mathematics and Computer Science  ``Ulisse Dini'' \\
University of Florence\\
Viale Morgagni 67/A \\
I-50134 Florence \\
Italy}
\email{daniele.mundici@unifi.it }

\date{\today}

\begin{document}

\keywords{AF algebra,    Elliott classification,    
Elliott local semigroup, Murray-von Neumann order, 
Grothendieck $K_0$ group, MV  algebra,  liminary C*-algebra,
 Effros-Shen algebra, Behnke-Leptin algebras,
Farey AF algebra.}

 \subjclass[2010]{Primary:   46L35.
Secondary:    
06D35,  06F20,   19A49,  46L80, 47L40.}

\begin{abstract} We assess the computational complexity
of several decision problems concerning  (Murray-von Neumann)
equivalence classes of projections of AF-algebras  whose
Elliott classifier is lattice-ordered. We construct polytime
reductions among many of these problems.
\end{abstract}

\maketitle

\section{introduction}

Let $ \mathfrak A$ be an AF-algebra, \cite{bra}.
The partial addition  + in  Elliott's  local
semigroup \cite{ell} of $\mathfrak A$   is uniquely 
extendible to a   total   operation
that preserves all algebraic and order 
properties of +  iff   the
Murray-von Neumann order $\preceq $  of 
 $\mathfrak A$  is a lattice,  for short,
$\mathfrak A$ is an {\it AF$\ell$-algebra}. 
See Theorem \ref{theorem:multi}.
AF$\ell$-algebras  have a preeminent role in the 
  literature on AF-algebras.
The Elliott classifier 
$E(\mathfrak A)$ of any AF$\ell$-algebra
$\mathfrak A$  is an {\it MV-algebra}, i.e., a Lindenbaum
algebra   in \luk\ logic \L$_\infty,  $ \cite{cigdotmun, mun11}.
All countable MV-algebras arise as  $E(\mathfrak A)$ 
for some AF$\ell$-algebra $\mathfrak A.$
Since $\mathfrak A$ is a quotient of the universal
AF$\ell$-algebra  $\mathfrak M$ of \cite[$\S 8$]{mun-jfa},
and $E(\mathfrak M)$ is the free MV-algebra over
countably many generators, each
  \L$_\infty$-formula  $\phi$ {\it naturally}
codes an equivalence class $\phi^\mathfrak A$ of projections in $\mathfrak A$.

We equip the Elliott monoid $E(\mathfrak A)$
of every AF$\ell$-algebra $\mathfrak A$
with a partial order $\phi^\mathfrak A
\sqsubseteq\psi^\mathfrak A$ intuitively meaning that all
projections in  $\phi^\mathfrak A$ are ``less eccentric''
than those in   $\psi^\mathfrak A$.
We prove that  $p$  is central  
iff  its  (always Murray-von Neumann)
 equivalence class  $[p]$ is  $\sqsubseteq$-minimal  iff 
it is a  (Freudenthal) 
characteristic element  of
$K_0(\mathfrak A)$ iff  $[p]\wedge [1-p]=0$.
Among others, we  consider the following decision   problems:
\,\,\, (a) $\phi^\mathfrak A\,\text{?`}
\, \hspace{-1mm}=?\,\,\psi^\mathfrak A$,   \,\, \,(b) 
  $\phi^\mathfrak A\,\text{?`}
  \hspace{-1mm}\preceq?\,\, \psi^\mathfrak A$,\,\,
  (c)   $\phi^\mathfrak A\,\text{?`}\hspace{-1mm}
  \sqsubseteq?\,\,\psi^\mathfrak A,$ \,\,
  (d)  $\phi^\mathfrak A\text{?`}\hspace{-1mm}=?\,0,$
 \,\,  (e) Is $\phi^\mathfrak A$  central?\,\,
We prove that problems
 (a)-(d) for 
 $\mathfrak A$
  have the same computational complexity, up to a polytime
  reduction.
 If  $\mathfrak A$ is simple, or if  $\mathfrak A$ has no
 quotient isomorphic to $\mathbb C$, then also
 problem  (e)  for $\mathfrak A$   has
 the same complexity as (a)-(d).
The complexity of each problem 
(a)-(e) is polytime for many
 relevant AF-algebras in the
literature, including 
the  Behncke-Leptin algebras  $\mathcal A_{m,n}$ and
every  Effros-Shen algebra  $\mathfrak F_\theta$, for 
$\theta\in \interval\setminus \mathbb Q$
  a real algebraic integer,  or for $\theta = 1/{\rm e}$. 
  
  For every AF$\ell$-algebra $\mathfrak A$ we let 
$\prim(\mathfrak A)$ denote the
set of primitive  (always closed and  two-sided)
  ideals of $\mathfrak A$ equipped   with the
Jacobson topology. By
\cite[3.8]{bra},  an ideal $\mathfrak P$
of  $\mathfrak A$ is primitive iff it 
is {\it prime}, in the
sense that 
whenever  ideals $\mathfrak P_1,
\mathfrak P_2$ of  $\mathfrak A$
satisfy  $\mathfrak P_1\cap \mathfrak P_2
= \mathfrak P$ then either $\mathfrak P_1
= \mathfrak P$ or $\mathfrak P_2= \mathfrak P$.

  For any MV-algebra $A$ we let
 $ \Spec(A)$
 denote the space of prime ideals of $A$ endowed with
 the Zariski (hull-kernel) topology,   (\cite[4.14]{mun11}). 
 As shown in   \cite[1.2.14]{cigdotmun}, 
 \begin{equation}
 \label{equation:subdirect} 
 \bigcap  \Spec(A)= \{0\}.
 \end{equation}
 
The basic properties of AF$\ell$-algebras are summarized by
the following two results:
 
\begin{theorem}
\label{theorem:multi}  
Let $\mathfrak A$ be an AF-algebra.
 
\smallskip 
{\rm(i)} \cite[Theorem 1]{munpan}
Elliott's partial addition $+$ in 
$L(\mathfrak A)$   has 
  at most one
 extension to an associative, commutative, 
monotone operation  $\oplus \colon
L(\mathfrak A)^2\to L(\mathfrak A)$ 
such that for each projection  $p \in  \mathfrak A$, 
$ [1_\mathfrak A - p ]$   
 is the smallest  $[q]\in L( \mathfrak A)$ satisfying
     $[p] \oplus [q]=[1_\mathfrak A].$ 
The  semigroup  $(S(\mathfrak A),\oplus)$
expanding  
$L(\mathfrak A)$  exists iff    $\mathfrak A$   
is an AF$\ell$-algebra.

\medskip
{\rm(ii)} \cite{ell}    Let    $\mathfrak A_1$   and    
$\mathfrak A_2$    be  AF$\ell$-algebras.    
For  each     $j = 1, 2$    let      $\oplus_j$    
be the  extension  of Elliott's addition  given by (i).    
Then the semigroups  $(S(\mathfrak A_1), \oplus_1)$    and   
 $(S(\mathfrak A_2), \oplus_2)$   are isomorphic iff so are      
$\mathfrak A_1$   and    
$\mathfrak A_2$.
 
\medskip
{\rm(iii)}   \cite[Theorem 2, Proposition 2.2]{munpan} 
$(S(\mathfrak A), \oplus)$ 
has  the richer  structure of  a  monoid
$(E(\mathfrak A), 0, ^*,  \oplus)$ 
with an involution   
$[p]^* =  [1_\mathfrak A - p ]$. 
The Murray-von Neumann  
lattice-order  of equivalence classes of projections 
    $[p],[q]\in E(\mathfrak A)$  is definable by  the
     involutive monoidal operations of $E(\mathfrak A)$ 
     upon setting
    $[p] \vee [q]=([p]^*\oplus [q])^*\oplus [q]$ and
     $[p]\wedge [q]= ([p]^*\vee [q]^*)^*.$   
We say that
$E(\mathfrak A)$ is the {\em Elliott monoid} of $\mathfrak A$.

\medskip
{\rm(iv)}  \cite[Theorem 3.9]{mun-jfa}  
The map 
 $\mathfrak A \mapsto  (E(\mathfrak A),
 0, ^*, \oplus)$ 
   is a  bijective correspondence
   between isomorphism classes of    AF$\ell$-algebras
        and  isomorphism classes of  
countable abelian monoids with a unary operation   $^*$  satisfying the equations:
$  x^{**}=x,  \,   0^* \oplus x =  0^*,   \,
\mbox{ and }   \,   (x^* \oplus y)^* \oplus
 y = (y^*\oplus x)^* \oplus x.$
The involutive monoids defined by  these equations are known as {\em MV-algebras},
\cite{cha, cigdotmun, mun11}.
Let\, $\Gamma$ be the categorical equivalence between
 unital  $\ell$-groups and MV-algebras defined in \cite[$\S 3$]{mun-jfa}.
Then 
 $(E(\mathfrak A),
 0, ^*, \oplus)$  is isomorphic to  
 $\Gamma(K_0(\mathfrak A),K_0(\mathfrak A)^+,[1_\mathfrak A] ))$.
 
\medskip
{\rm(v)}{ \rm (\cite{dav, eff, goo-shiva})}\,\,
For any
   AF$\ell$-algebra 
 $\mathfrak A$
the dimension group 
$$K_0(\mathfrak A)=(K_0(\mathfrak A),K_0(\mathfrak A)^+,[1_\mathfrak A] )$$
is a countable lattice-ordered abelian group
with a distinguished strong order unit (for short,
a {\em unital $\ell$-group}).
All countable unital $\ell$-groups arise in this way.
Let    $\mathfrak A_1$   and    
$\mathfrak A_2$    be  AF$\ell$-algebras.
Then $K_0(\mathfrak A_1)$ and $K_0(\mathfrak A_2)$
  are isomorphic  unital  
  $\ell$-groups  iff\,\,       
$\mathfrak A_1$   and    
$\mathfrak A_2$ are isomorphic.
\hfill{$\Box$}
 \end{theorem} 

\begin{theorem}
\label{theorem:spectral}
 In any AF$\ell$-algebra $\mathfrak A$ we have:
 
\medskip
{\rm(i)} The map
 $
\eta \colon   \mathfrak I   
 \mapsto K_0(\mathfrak I)\cap E(\mathfrak A)   
 $
 is an isomorphism of
  the  lattice of  ideals of $\mathfrak A$ onto the lattice of ideals
 of    $E(\mathfrak A)$. Primitive ideals of $\mathfrak A$ correspond via $\eta$ to 
  prime ideals of $E(\mathfrak A)$.
  
\medskip
{\rm(ii)}  In particular,   $\eta$ is  a
 homeomorphism  of  
 $\prim(\mathfrak A)$  
 onto  
$
 \Spec(E(\mathfrak A))
$.

\medskip
{\rm(iii)} Suppose $\mathfrak j$ is an
 ideal of the countable MV-algebra $A$.  Let
 the AF$\ell$-algebra  $\mathfrak A$ be  defined   
 by $E(\mathfrak A)=A$, as in  Theorem \ref{theorem:multi}(iv). 
 Let  the ideal $\mathfrak J$  of $\mathfrak A$ be defined by
 $\eta(\mathfrak J)= \mathfrak j$.  Then   
 $A/\mathfrak j$ is isomorphic to 
 $E(\mathfrak A/\mathfrak J).$

\medskip
{\rm(iv)}
 For  every ideal $\mathfrak I$ of $\mathfrak A$,  the map
 $
 \left[{p}/{\mathfrak I}\right]\mapsto {[p]}/{\eta(\mathfrak I)}\,,$
 (with $p$  ranging over all projections in $\mathfrak A$), 
is an isomorphism
of $E(\mathfrak A/\mathfrak I)$ onto $E(\mathfrak A)/\eta(\mathfrak I).$ 
In particular, for every
$\mathfrak P\in \prim(\mathfrak A)$ the
 MV-algebra $E(\mathfrak A/\mathfrak P)$
is   totally ordered  and  
 $\mathfrak A/\mathfrak P$ 
  has comparability of projections
 in the sense of Murray-von Neumann.  
 
\medskip
{\rm(v)} The map
 $
 \mathfrak l
 \mapsto \mathfrak l\, \cap \, \Gamma(K_0(\mathfrak A))
 $
 is an isomorphism of the lattice of
  {\em ideals} of $K_0(\mathfrak A)$
  (i.e., kernels of unit preserving
 $\ell$-homomorphisms of 
 $K_0(\mathfrak A)$) 
 onto the lattice of ideals of $E(\mathfrak A).$  Further, 
 $$
 \Gamma\left(\frac{K_0(\mathfrak A)}{\mathfrak l}\right)
 \,\,\cong\,\,
  \frac{\Gamma(K_0(\mathfrak A))}{\,\,\mathfrak l\,\, \cap \,\,\Gamma(K_0(\mathfrak A))\,\,}\,. 
 $$
\end{theorem}

\smallskip
\begin{proof} (i)\,\, From   \cite[Proposition IV.5.1]{dav} and
 \cite[p.196, 21H]{goo-shiva} one gets an isomorphism
of  the lattice of ideals of $\mathfrak A$ onto the lattice of
 ideals of the unital $\ell$-group $K_0(\mathfrak A).$ 
The
preservation properties of $\Gamma$, \cite[Theorems 7.2.2, 
 7.2.4]{cigdotmun},  then yield the desired isomorphism.
 The second statement immediately follows from
  \cite[Theorem 3.8]{bra} and the above mentioned
   characterization of
  prime ideals in MV-algebras,
  (\cite[Proposition 4.13]{mun11}).

\smallskip
(ii)  follows from (i), by definition of the topologies
of $\prim(\mathfrak A)$  and   $\Spec(E(\mathfrak A))$.

\medskip
(iii)  \,  We have an exact sequence
$
0\to \mathfrak J \to \mathfrak A\to \mathfrak A/\mathfrak J\to 0.
$
Correspondingly (\cite[IV.15]{dav} or 
 \cite[Corollary 9.2]{eff}),  we have an
exact sequence
$
0\to K_0(\mathfrak J) \to K_0(\mathfrak A)\to 
K_0(\mathfrak A/\mathfrak J) \to 0,
$
whence the unital $\ell$-groups 
$
K_0\left({\mathfrak A}/{\mathfrak J}\right)
$ 
and
$
{K_0(\mathfrak A)}/{K_0(\mathfrak J)}
$
are isomorphic. 
The preservation properties of $\Gamma$ under quotients
 \cite[Theorem 7.2.4]{cigdotmun},
  together with Theorem \ref{theorem:multi}(iv)-(v) yield
  $$
E\left(\frac{\mathfrak A}{\mathfrak J}\right)
\cong  \Gamma\left(K_0\left(\frac{\mathfrak A}{\mathfrak J}\right)\right)\cong
   \frac{\Gamma(K_0(\mathfrak A))}{K_0(\mathfrak J)\cap
   \Gamma(K_0(\mathfrak A))}\cong\frac{E(\mathfrak A)}
   {\eta(\mathfrak J)}=
   \frac{A}{\mathfrak j}.
  $$
  
\smallskip
(iv)  Combine  (i) and  (iii)  with the preservation properties of
$K_0$ for exact sequences  and the preservation
properties of $\Gamma$ under quotients. 
The   MV-algebra
$E(\mathfrak A/\mathfrak P)$
is   totally ordered. As a matter of fact,  by (ii),
$\eta(\mathfrak P)$ belongs to $ \Spec(E(\mathfrak A))$ 
whenever  $\mathfrak P$ beongs
to $ \prim(\mathfrak A)$.
By Theorem \ref{theorem:multi}(iv), 
 $\mathfrak A/\mathfrak P$ 
  has comparability of projections.  
 
\smallskip
(v) This again follows from 
  \cite[Theorems 7.2.2, 7.2.4]{cigdotmun}.  
 \end{proof}

\section{Characterizing central projections in AF$\ell$-algebras}
\label{section:central}
Following
\cite[p.22]{ghl},
by a {\it characteristic element} in a partially ordered
abelian group $G$ with order unit $u$, we mean an
element $c$ with $0\leq c\leq u$ such that the greatest lower bound
of the set $\{c,  u-c\}$ exists and 
 equals 0. 
In the framework  of real vector lattices,
 with  $u$ an arbitrary positive element, $c$ is known as a
  ``component of $u$'',
\cite[p.13]{alibur}, \cite[p.284]{luxzaa}. This 
 notion goes back to Freudenthal, \cite[p. 643, (5.2)]{fre}.

\smallskip 
In every MV-algebra $A$ one defines
$a\odot b=(a^*
\oplus b^*)^*\,\,\, \mbox{ and }\,\,\,
\dist(a,b)=(a\odot b^*)\oplus(b\odot  a^*).$
Chang \cite{cha}
 writes $x\cdot y$ instead 
 of $x\odot y$, and  $d(x,y)$ instead of $\dist(x,y)$. 
 The latter is known as {\it Chang's distance function.}
See \cite[p.8 and 1.2.4]{cigdotmun} and 
  \cite[p.477]{cha}.
Further, let us use the notation 
$
\Boole(A) 
$
for the  set of {\em idempotent}
elements of the MV-algebra  $A$,
\begin{equation}
\label{equation:boole}
\Boole(A) =\{a\in A\mid a\oplus a =a\}.
\end{equation}
As observed by Chang  \cite[Theorems 1.16-1.17]{cha},
${\Boole}(A)$  
  is a subalgebra of $A$ which turns out to be
a boolean algebra.
By an {\it MV-chain} we mean an MV-algebra  $D$
whose underlying
order is total, \cite[p.10]{cigdotmun}.

\medskip 
\subsection*{The binary relation 
``$p$ is closer than $q$ to the center''}
Let $\mathfrak P$ be a primitive ideal of an AF$\ell$-algebra
$\mathfrak A.$ By Theorem \ref{theorem:spectral} the
quotient  AF$\ell$-algebra  $\mathfrak A/\mathfrak P$ has
comparability of projections. Thus for any 
projections 
$p,q\in \mathfrak A$ we have the following three
mutually incompatible cases:  
$$p/\mathfrak P\prec q/\mathfrak P, \,\,
p/\mathfrak P\succ q/\mathfrak P,\,\,
p/\mathfrak P \sim  q/\mathfrak P.$$
In particular, we have incompatible cases:  
$$p/\mathfrak P\prec p^\ast/\mathfrak P, \,\,
p/\mathfrak P\succ p^\ast/\mathfrak P,\,\,
p/\mathfrak P \sim  p^\ast/\mathfrak P.$$

\smallskip
\noindent
The proof of the following result   will appear
elsewhere:

\begin{theorem}
[Characterization of central projections]
\label{theorem:threeconditions}
Suppose $\mathfrak A$ is
an   AF$\ell$-algebra. 
In view of Theorem
\ref{theorem:multi}(iv)-(v),  let us identify
$E(\mathfrak A)$ with the unit interval
$\Gamma(K_0(\mathfrak A))$ of the
unital $\ell$-group
$K_0(\mathfrak A).$

\smallskip
\noindent
(I)
For  every projection $p$ of $\mathfrak A$ 
the  following conditions are equivalent:
\begin{itemize}
\item[(i)] $p/\mathfrak P \in\{0,1\}$,   (the trivial 
elements of  $\mathfrak A/\mathfrak P$) \,
 for all \,\,$\mathfrak P\in \prim(\mathfrak A)$.

\smallskip 
\item[(ii)]  $[p]\in {\Boole}(E(\mathfrak A))$.

\smallskip  
\item[(iii)]  $p$ is central in $\mathfrak A$.

\smallskip  
\item[(iv)] $[p]$ is a  characteristic element  of
$K_0(\mathfrak A).$ 
\end{itemize}

\smallskip
\noindent
(II)
For any  $x,y\in E(\mathfrak A)$ let us write
$
 x\sqsubseteq y
$
 \,\, iff \,\, 
  for every  prime ideal $\mathfrak p$ of $E(\mathfrak A)$
$$
(y/\mathfrak p<y^*/\mathfrak p\,\,\,\mbox{\it implies}
  \,\,\,  x/\mathfrak p\leq y/\mathfrak p)
    \,\,\,\,\mbox{  and  }\,\,\,
 (y /\mathfrak p>y^*/\mathfrak p \,\,\,
\mbox{\it implies}  \,\,\,   x /\mathfrak p\geq y/\mathfrak p),
$$
with $\leq$ the underlying  total order of 
$E(\mathfrak A)/\mathfrak p$  defined 
in Theorem \ref{theorem:multi}(iii).
Then 
$\sqsubseteq$ endows 
$E(\mathfrak A)$
with  a partial order  relation.  
Moreover, for 
 every projection  $p$ in $\mathfrak A$,
the equivalent 
the   conditions  
   in (I) 
are also equivalent to
$[p]$ being  $\sqsubseteq$-minimal in $E(\mathfrak A).$
\end{theorem}

\medskip
 \section{The algorithmic theory of AF$\ell$-algebras} 
 \label{section:algorithmic}

   Our standard
 reference for computability theory is \cite{macyou}.
 For  \L$_\infty$  
  we refer to \cite{cigdotmun,mun11}.

 \subsection*{The syntax of 
 $\mathsf{TERM}_n$ and $\mathsf{TERM}_\omega$}
The set   $\{0,^*,\oplus\}$  of constant and
 operation {\it symbols}  of  involutive monoids is enriched
by adding   countably many variable symbols  $X_1,X_2,\dots$.
Henceforth, the set
$\mathcal A= \{0,^*,\oplus, 
X_1,X_2,\dots, ),(, \}$
is our {\it alphabet}.
Parentheses are added to construct a non-ambiguous
readable syntax.
The set  $\mathcal A^*$ of {\it strings} over $\mathcal A$
is defined by 
$
\mathcal A^*=  \{(s_1,s_2,...,s_l)\in \mathcal  A^l\mid
l=0, 1,2,\dots \}.
$
Let $n=1,2,\dots$.
By a {\it term  
 in the variables
$X_1,\dots,X_n$}\/,
we mean a  string $\phi\in \mathcal A^*$ obtainable
by the following inductive
definition:
\,\,\,(*)  0 and $X_1,X_2,\dots,X_n$ 
are terms;\,\,\, (**)  if  $\alpha$ and $\beta$ are terms,
then so are $\alpha^*$ and $(\alpha\oplus\beta).$
We let 
 $\mathsf{TERM}_n$  be   the set of terms
in the variables $X_1,X_2,\dots,X_n$.
Elements of  $\mathsf{TERM}_n$ 
are also known as $n$-variable {\it \luk\ formulas}.
We also let 
 $\mathsf{TERM}_\omega=\bigcup_n\mathsf{TERM}_n.$

 \subsection*{Coding  equivalence classes
 of projections by \L$_\infty$-formulas}
Fix   $n=1,2,\dots $.
By McNaughton's theorem \cite{mcn},
the coordinate functions  $\{\pi_1,\dots,\pi_n\}$ are
a distinguished free generating set of the free
MV-algebra  $\McNn\cong E(\mathfrak M_n).$
In view of Theorem
\ref{theorem:multi}(iv),  the   AF$\ell$-algebra
$\mathfrak M_n$ is defined by 
$
E(\mathfrak M_n)=\McNn.
$

Arbitrarily pick representative projections
$\mathsf p_1,\dots, \mathsf p_n\in \mathfrak M_n$
such that the equivalence classes 
$[\mathsf p_1],\dots, [\mathsf p_n]\in E(\mathfrak M_n)$
correspond to  $\pi_1,\dots,\pi_n$ via Elliott's classification.
We
may naturally say that the variable symbol $X_i$
{\it codes} both the coordinate function
  $\pi_i\in \McNn$ and the equivalence class $[\mathsf p_i]
\in E(\mathfrak M_n). $  For definiteness, we will 
say that  $\pi_i$ is   the {\it interpretation of  $X_i$ in
(the Elliott monoid of)
$\mathfrak M_n$}, in symbols,
$
X_i^{\mathfrak M_n} =  \pi_i.
$
For every  $\phi\in \mathsf{TERM}_n,$
 the {\it interpretation
  $\phi^{\mathfrak M_n}$ of $\phi$ in 
  (the Elliott monoid of) $\mathfrak M_n$}
is then defined by 
(*)  $ 0^{\mathfrak M_n} = \mbox{the constant zero function over }\interval^n$
  and  inductively,
(**)  
$(\psi^*){^{\mathfrak M_n}}= (\psi^{\mathfrak M_n})^*,\,\,\,
 {(\alpha \oplus \beta)}^{\mathfrak M_n}
  =  (\alpha^{\mathfrak M_n}\oplus\beta^{\mathfrak M_n}).$
We  also   say that $\phi$ {\it codes}  
$\phi^{\mathfrak M_n}$
in (the Elliott monoid of) $\mathfrak M_n$.
By a traditional abuse of notation,   the MV-algebraic  
 operation {\it symbols} also denote  
 their corresponding {\it operations}. 
More generally,  let $\mathfrak A=\mathfrak M_n/\mathfrak I$
 be an
AF$\ell$-algebra,  
 for some
ideal $\mathfrak I$ of $\mathfrak M_n.$
Let 
  \begin{equation}
  \label{equation:ore1} 
  \mathfrak i=K_0(\mathfrak I)\cap E(\mathfrak M_n).
  \end{equation}
  be the ideal of $\McNn$ corresponding to $\mathfrak I$
by  Theorem \ref{theorem:spectral}(v).
Via the  identification 
 $E(\mathfrak M_n/\mathfrak I) =
 \McNn/\mathfrak i,$
  the 
{\it interpretation
  $\phi^{\mathfrak A}$ of $\phi$ in (the Elliott monoid of)
   $\mathfrak A$} 
  is defined by:  $0^\mathfrak A=\{0\}\subseteq \mathfrak A$,
  %
$X_i^{\mathfrak A}\,\,\,=\,\,\,X_i^{\mathfrak M_n/\mathfrak I}
\,\,\,=\,\,\,
X_i^{\mathfrak M_n}/\mathfrak i \,\,\,=\,\,\, \pi_i/\mathfrak i$
and inductively, 
$  (\psi^{*})^{\mathfrak A}= (\psi^{\mathfrak A})^*,\,\,\,\,\,\,\,
 {(\alpha \oplus \beta)}^{\mathfrak A}
  =  (\alpha^{\mathfrak A}\oplus\beta^{\mathfrak A}).$
The following identities are easily verified by induction on the
syntactical complexity of terms  (i.e., the number of symbols
occurring in each term), using the unique readability property
of the syntax:
 $(\psi^{*})^{\mathfrak A} = \frac{(\psi^{*})^{\mathfrak M_n}}
 {\mathfrak i}
 =\left(\frac{\psi^{\mathfrak M_n}}{\mathfrak i}\right)^*
 = \frac{(\psi^{\mathfrak M_n})^*}{\mathfrak i}$
 and
 $ {(\alpha\oplus \beta)}^{\mathfrak A}=
 \frac{(\alpha\oplus \beta)^{\mathfrak M_n}}{\mathfrak i}=
\frac{\alpha^{\mathfrak M_n}}{\mathfrak i}\oplus
\frac{\beta^{\mathfrak M_n}}{\mathfrak i}=
\frac{\alpha^{\mathfrak M_n}\oplus
 \beta^{\mathfrak M_n}}{\mathfrak i}.$
We also  say that  $\phi$   {\it codes  $\phi^\mathfrak A$ in 
(the Elliott monoid of) $\mathfrak A$}.

 \medskip
 The interpretation  $\phi^\mathfrak M$ of
 $\phi\in \mathsf{TERM}_\omega$ 
  in   (the Elliott monoid of)  the  universal AF$\ell$-algebra
 $\mathfrak M$ and its quotients is similarly defined.

\subsection*{Decision problems for  AF$\ell$-algebras}
Fix a cardinal $\kappa=1,2,\dots,\omega$. Let   
 $\mathfrak A=\mathfrak M_\kappa/\mathfrak I$
 for some ideal $\mathfrak I$ of $\mathfrak M_\kappa.$
 (Here  $\mathfrak M_\omega=\mathfrak M$.)
The  {\it word  problem  $\mathsf P_{1}$   
of  $\mathfrak A$} is defined  by 
$
\mathsf P_{1}=
 \{(\phi,\psi)\in \mathcal A^{*} \mid
 (\phi,\psi)\in \mathsf{TERM}_\kappa^2
 \mbox{ and }
 \phi^\mathfrak A= \psi^\mathfrak A \}.
$
Intuitively,   
on input strings
 $\phi$ and $\psi$,  
	\,\,  $\mathsf P_{1}$
checks  if   $\phi$ and $\psi$ are  strings in $\mathsf{TERM}_\kappa$\,\,\,
coding  the same equivalence
class of projections of $\mathfrak A$.
  Likewise, the  {\it order  problem}  $\mathsf P_{2}$
   of $\mathfrak A$
   is the subset 
 of $\mathcal A^{*}$
given by 
$ \{(\phi,\psi)\in \mathsf{TERM}_\kappa^2\mid
 \phi^\mathfrak A \preceq \psi^\mathfrak A \}.$ 
Problem $\mathsf P_{2}$
 checks if   $\phi$   codes  an equivalence
class of projections  $ \phi^\mathfrak A$ in $\mathfrak A$  
that  precedes $ \psi^\mathfrak A$
in the Murray-von Neumann order $\preceq$ of 
projections in $\mathfrak A$.
 The  {\it  eccentricity  problem}  $\mathsf P_{3}
 = \{(\phi,\psi)\in \mathsf{TERM}_\kappa^2\mid
 \phi^\mathfrak A \sqsubseteq \psi^\mathfrak A \}$
   of $\mathfrak A$  checks whether $\phi$ codes  an equivalence
class of projections  $ \phi^\mathfrak A$ in $\mathfrak A$  
that  precedes $ \psi^\mathfrak A$
in the  $\sqsubseteq$-order of $\mathfrak A$.
  Further,  the   {\it zero  problem}  $\mathsf P_{4}
 = \{\phi \in \mathsf{TERM}_\kappa\mid
 \phi^\mathfrak A =0  \}$
   of $\mathfrak A$  checks if  $ \phi^\mathfrak A =0.$
 The  {\it central  projection problem}  $\mathsf P_{5}$
   of $\mathfrak A$  
   checks if  $\phi^\mathfrak A $ is an equivalence class of  
 central projections in $\mathfrak A$. 
   The  {\it nontrivial projection  problem}   $\mathsf{P}_6$   
of   $\mathfrak A$  
checks if  $ \phi^\mathfrak A$
  different from 0 and 1.
The  {\it  central nontrivial projection  problem}  $\mathsf{P}_7$ 
of  $\mathfrak A$  
checks if $ \phi^\mathfrak A$ is an equivalence class
of   central
 projections of $\mathfrak A$ other than 0 or 1.

  \subsection*{Polytime  problems and  (Turing) reductions}
 For any formula  $\phi\in \mathsf{TERM}_\kappa$ we let
$
 \mbox{$ ||\phi||$ =  the number of occurrences of symbols in  
$ \phi.$}
$
 Let $i=1,\dots,7$. We say that $\mathsf{P_i}$
 is {\it decidable in polytime}  if 
  there is
  a polynomial $q\colon\{0,1,2,\dots\}\to \{0,1,2,\dots\}$ 
  and a Turing machine  $\mathcal T$ with the
  following property:
Over any input
string $\sigma\in \mathcal A^*$, machine $\mathcal T$
decides if $\sigma$ belongs to 
$\mathsf{P_i}$ within a number of steps  $\leq q(||\sigma||)$.
 Given problems
 $\mathsf{P'}, \mathsf{P''}\subseteq \mathcal A^*$, a
   {\it reduction} $\rho$ of $\mathsf{P'}$ to $\mathsf{P''}$
 is a map 
 $
 \rho\colon \mathcal A^*\to \mathcal A^*
 $
 such that for every string $\sigma\in \mathcal A^*$,
 \,\,\,  $\sigma\in  \mathsf{P'}$\,\,\, iff\,\,\,  
 $\rho(\sigma) \in  \mathsf{P''}.$
 When $\rho$ is computable in polynomial time
 we say it is a {\it polytime reduction.}
Compare with 
  \cite[p.211]{macyou}.

\medskip

\begin{proposition}
\label{proposition:zhomeo}
Let  $\mathfrak A$ be an  AF$\ell$-algebra whose Elliott
monoid  is finitely generated.  
Let  $\mathfrak I$ (resp., $\mathfrak J$) be an 
ideal of   $\mathfrak M_m$ 
(resp.,   of $\mathfrak M_n$) such that
$\mathfrak A\cong \mathfrak M_m/\mathfrak I
\cong \mathfrak M_n/\mathfrak J.$
Let $\mathsf{P}\subseteq \mathcal A^*$ be any problem
among $\mathsf{P}_1,\dots,\mathsf{P}_7$.
Then  
  $\mathsf{P}$
 for   $\mathfrak M_m/\mathfrak I$ is polytime reducible 
  to $\mathsf{P}$ for   $\mathfrak M_n/\mathfrak J.$
Thus in particular,  problem $\mathsf{P}$
for  $\mathfrak M_m/\mathfrak I$ 
    is decidable in polytime
iff so is $\mathsf{P}$
for  $\mathfrak M_n/\mathfrak J$.
\end{proposition}

\begin{proof}
A  {\it rational polyhedron}
$P\subseteq [0,1]^{n}$
is  the union of finitely many simplexes 
$T_i\subseteq [0,1]^{n}$ with  rational vertices.
Two  rational polyhedra $P$ and $Q$
are   {\it $\mathbb Z$-homeomorphic}
 if $Q$ is the image of $P$
under a piecewise linear homeomorphism $\eta$ 
such that all linear pieces of $\eta$ and $\eta^{-1}$ 
have integer coefficients. 
Theorem \ref{theorem:spectral}
yields  ideals $\mathfrak i,\mathfrak j$
such that 
$E(\mathfrak A) \cong \McNm/\mathfrak i
\cong \McNn/\mathfrak j.$
Elliott's
classification, and the 
  $\mathbb Z$-homeomorphism
in  \cite[8.7]{mun11}, (coded by an $m$-tuple $\vec\psi_\eta$ of
$n$-variable terms)   
yields a  polytime  reduction 
$\rho'\colon \phi\in \mathsf{TERM}_m\mapsto
 \phi(\vec\psi_\eta)
\in \mathsf{TERM}_n$ of
problem $\mathsf P_i$  for 
$ \mathfrak M_m/\mathfrak I$ 
to  $\mathsf P_i$ for   $\mathfrak M_{n}/\mathfrak J$,
and a polytime reduction $\rho''$ in the opposite direction. 
\end{proof}

\begin{theorem}
\label{theorem:mutual}
Fix an  AF$\ell$-algebra $\mathfrak A$ whose 
 Elliott monoid is  finitely generated.
Then for  all    $i,j\in  \{1,\dots,4\}$
problem     $\mathsf{P}_i$ for  $\mathfrak A$
is polytime reducible to problem 
  $\mathsf{P}_j$.
Thus in particular,  $\mathsf{P}_i$ for $\mathfrak A$  
is polytime  
 iff so is  $\mathsf{P}_j$ for $\mathfrak A$.
\end{theorem}

\begin{proof} 
Following \cite[p.8]{cigdotmun},  for every MV-algebra
$A$ we define the operation $\ominus\colon A^2\to A$ by  
$x\ominus y=
x\odot y^*.$
 By   \cite[Lemma 1.1.2]{cigdotmun},
\begin{equation}
\label{equation:order}
\mbox{$x\leq y$\,\,\, iff\,\,\,
$x\ominus y=0.$}
\end{equation}
For all input terms $\alpha, \beta$, let us write for short
$\alpha^{\mathfrak A}=a$ and $\beta^{\mathfrak A}=b.$
By definition of the $\sqsubseteq$-order,\,\,\,
  $a\sqsubseteq b$\,\,\, iff\,\,\,
 for every  prime ideal $\mathfrak p$ of 
 the MV-algebra $E(\mathfrak A)$ the folllowing holds: 
 \begin{equation}
 \label{equation:dmanga}
 (b/\mathfrak p<b^*/\mathfrak p \,\,\mbox{implies}\,\,   a/\mathfrak p\leq 
 b/\mathfrak p)
\,\,\,\,\,\, \mbox{and}\,\,\,\,\,\,  
 (b /\mathfrak p>b^*/\mathfrak p  \,\,\mbox{implies}\,\,   a /\mathfrak p
 \geq b/\mathfrak p).
 \end{equation}
This is equivalent to 
$  (b/\mathfrak p\leq  b^*/\mathfrak p\,\,\,\,\mbox{and} \,\,\,\,  
  a/\mathfrak p\leq b/\mathfrak p)
\,\,\,\,\mbox{or}\,\,\,\, 
 (b /\mathfrak p\geq  b^*/\mathfrak p \,\,\,\,\mbox{and}\,\,\,\,   
 a /\mathfrak p\geq b/\mathfrak p),
 $
 which by \eqref{equation:order} can be reformulated as 
 \begin{equation}
 \label{equation:martede}
  (b/\mathfrak p\ominus  b^*/\mathfrak p=0=  a/\mathfrak p\ominus
   b/\mathfrak p)\, \,\,\mbox{or} \,\,\,  
 (b^*/\mathfrak p\ominus b /\mathfrak p=0=    b/\mathfrak p\ominus a /\mathfrak p).
\end{equation}
By definition of the lattice operations in the MV-chain   
 $E(\mathfrak A)/\mathfrak p$, (Theorem \ref{theorem:multi}(iii)),
 \eqref{equation:dmanga}-\eqref{equation:martede}
 equivalently state 
 $
 ((b/\mathfrak p\ominus  b^*/\mathfrak p) \vee  (a/\mathfrak p
 \ominus b/\mathfrak p)) \wedge 
((b^*/\mathfrak p\ominus b /\mathfrak p) \vee   
(b/\mathfrak p\ominus a /\mathfrak p))=0.
$
By  \eqref{equation:subdirect}, the  Elliott monoid
$E(\mathfrak A)$ satisfies: 
 $
  a\sqsubseteq b \,\,\,\mbox{iff}\,\,\, 
  ((b \ominus  b^* ) \vee  (a \ominus b )) \wedge 
((b^* \ominus b  ) \vee   (b \ominus a  ))=0.
$
By
 definition of $a,b$, 
$
 a\sqsubseteq b
\,\,\,\mbox{iff}\,\,\,  (((\beta \ominus  \beta^* ) \vee  (\alpha \ominus \beta )) \wedge 
((\beta^* \ominus \beta  ) \vee 
  (\beta \ominus \alpha  )))^{\mathfrak A}=0.
$
We have just constructed 
 a polytime reduction of the eccentricity
 problem  $\mathsf{P}_3$ for  $\mathfrak A$ to the zero problem
 $\mathsf{P}_4$ for $\mathfrak A$. 
A converse polytime reduction is immediately obtained by
noting that 
$ \phi^\mathfrak A=[0]$ iff $ \phi^\mathfrak A\sqsubseteq [0]$,
whence $\mathsf{P}_4$ is a subproblem of 
$\mathsf{P}_3$.

\smallskip
   A polytime reduction of the order problem
   $\mathsf{P}_2$ to 
the zero problem $\mathsf{P}_4$
 for $\mathfrak A$  trivially follows by noting that  
  $ \phi^\mathfrak A\preceq  \psi^\mathfrak A$ iff  
$(\phi \ominus \psi)^\mathfrak A =0$.
Conversely, 
a polytime reduction of $\mathsf{P}_4$ to 
$\mathsf{P}_2$ for $\mathfrak A$   follows by noting that
  $ \phi^\mathfrak A=0$ iff $ \phi^\mathfrak A\preceq 0$.

\smallskip
    A polytime reduction of the word problem
    $\mathsf{P}_1$ to 
the zero problem $\mathsf{P}_4$ for $\mathfrak A$  follows from 
  $ \phi^\mathfrak A= \psi^\mathfrak A$ iff 
  $\dist( \phi^\mathfrak A, \psi^\mathfrak A)=0$
  iff  $((\phi\ominus \psi)\oplus(\psi\ominus \phi))^\mathfrak A=0.$
Conversely, a polytime reduction of  
$\mathsf{P}_4$  to $\mathsf{P}_1$
  for $\mathfrak A$  
 trivially  follows from the fact that
 the former problem
is a special case of the latter.
\end{proof}

Concerning the central projection problem  $\mathsf{P}_5$
we have:

  \begin{theorem}
\label{theorem:semi-mutual}
Arbitrarily fix an AF$\ell$-algebra $\mathfrak A$
whose Elliott monoid 
$E(\mathfrak A)$ is finitely generated.

\smallskip
(i)  There is a polytime reduction of the central
projection  problem  $\mathsf{P}_5$ to the zero problem
$\mathsf{P}_4$  for $\mathfrak A.$

\smallskip
(ii) The converse polytime reduction exists if 
$\mathfrak A$ has no quotient isomorphic  to $\mathbb C$,
or if  $\mathfrak A$ is simple.
\end{theorem}

\begin{proof}
(i) As a matter of fact, with the notation of \eqref{equation:boole},
we have:
\begin{align*} 
 &\,\,\, \phi^{\mathfrak A}\,\, \mbox{is 
 the   equivalence class of a central 
 projection  in  $\mathfrak A$}\\
 \,\,\mbox{iff}\,\,&\,\,\,
\phi^{\mathfrak A}
  \mbox{ 
belongs to}\, 
 \Boole(E(\mathfrak A)), \mbox{ by  Theorem  \ref{theorem:threeconditions}}\\
  \,\,\mbox{iff}\,\,&\,\,\,
\phi^{\mathfrak A}\wedge (\phi^{\mathfrak A})^*=0,
  \mbox{   by \cite[Theorem 1.16]{cha},
  \,\,(with  $\wedge$ as in Theorem
  \ref{theorem:multi}(iii))}\\
    \,\,\mbox{iff}\,\,&\,\,\,
 (\phi \wedge  \phi^*)^{\mathfrak A}=0.
\end{align*}
The desired polytime reduction
transforms  $\phi$ into   $\phi \wedge  \phi^*$.

\bigskip
(ii) Let us first assume 
\begin{equation}
\label{equation:first-assumption}
\mbox{
$\mathfrak A$ has no quotient isomorphic to $\mathbb C$.}
\end{equation}
Pick an  $n=1,2,\dots$ and an ideal
$\mathfrak I$ of  $\mathfrak M_n$ such that
$\mathfrak A\cong \mathfrak M_n/\mathfrak I.$
By Proposition  \ref{proposition:zhomeo},
our proof will not depend on the actual choice
of $n$ and $\mathfrak I.$
Let $A=E(\mathfrak A)$ and   $\mathfrak i$ be the ideal of
$E(\mathfrak M_n)$ corresponding to $\mathfrak I$
by  Theorem \ref{theorem:spectral}(v),
%
%
$\mathfrak i= K_0(\mathfrak I)\cap E(\mathfrak M_n).$
We have 
$E(\mathfrak A)=E(\mathfrak M_n/\mathfrak I) =
E(\mathfrak M_n)/\mathfrak i = \McNn/\mathfrak i= A.$
Setting 
 $
 X_i^\mathfrak A= \pi_i/\mathfrak i, \,\,\,\,(i=1,\dots,n),
 $
 every $\phi\in \mathsf{TERM}_n$ is interpreted
as
$\phi^{\mathfrak A}= \phi^{\mathfrak M_n}/\mathfrak i.$

 \medskip
 \noindent
 To obtain the desired polytime  reduction, let us
 set
\begin{equation}
\label{equation:rho}
 \rho\colon \phi\mapsto
 \phi\wedge \bigvee_{i=1}^n (X_i\wedge X_i^*),\,\,\,\,\,\,
 \,\,\,\mbox{for each $\phi\in \mathsf{TERM}_n$}.  
\end{equation}
We must prove 
 \begin{equation}
 \label{equation:three}
\phi^{\mathfrak A}=0
 \Leftrightarrow \left(\phi\wedge \bigvee_{i=1}^n
 (X_i\wedge X_i^*)  \right)^\mathfrak A \in \{0,1\}.
 \end{equation}
 Indeed, by Theorem
 \ref{theorem:threeconditions},
 equation
  \eqref{equation:three} amounts to saying that
 $ \phi^{\mathfrak A}=0$ iff the term  $\rho(\phi)$
 given by 
    $\phi\wedge \bigvee_{i=1}^n (X_i\wedge X_i^*)$ 
    codes    the
      equivalence class of a central projection of 
  $\mathfrak A.$

The $(\Rightarrow)$   direction of   \eqref{equation:three}
  is trivial.  For the  $(\Leftarrow)$ direction,
  arguing by way of contradiction, let us assume
    \begin{equation}
 \label{equation:five=AH1}
 \left(\phi\wedge \bigvee_{i=1}^n
 (X_i\wedge X_i^*)  \right)^\mathfrak A
 \subseteq \{0,1\}
 \,\,\,\, \mbox{ and }\,\,\,\,
 \phi^{\mathfrak A} \not=0.
 \end{equation}

 \noindent
 By \eqref{equation:subdirect}, 
   for  some  prime ideal $\mathfrak p$ of $A$
  we have
 \begin{equation}
 \label{equation:six}
\phi^{\mathfrak A}/\mathfrak p\not=0
\mbox{  in the MV-chain  $  A/\mathfrak p$ }.
 \end{equation}
 On the other hand, every   prime ideal $\mathfrak q$
 satisfies 
${   \left(\phi\wedge \bigvee_{i=1}^n
 (X_i\wedge X_i^*)  \right)^\mathfrak A}/{\mathfrak q}\in \{0,1\}$
  in the MV-chain  
 $A/\mathfrak q$.
 In particular,
   \begin{equation}
  \label{equation:eight}
  \frac{ \left(\phi\wedge \bigvee_{i=1}^n
 (X_i\wedge X_i^*)  \right)^\mathfrak A}{\mathfrak p}\in \{0,1\}
 \mbox{  in the MV-chain  $  A/\mathfrak p$ }.
 \end{equation}
 
 \medskip
 \noindent
Another application of
  Theorem \ref{theorem:spectral} yields  a unique
 primitive ideal $\mathfrak P$ of $\mathfrak A$ such that
 $E(\mathfrak A/\mathfrak P)=E(\mathfrak A)/\mathfrak p$

\medskip 
 \noindent
 {\it Claim:} For each $i=1,\dots,n$ the interpretation
  $X_i^\mathfrak A/\mathfrak p$
  of the variable symbol $X_i$
 in   $\mathfrak A/\mathfrak P$, as well as in
 its Elliott monoid  
$\in A/\mathfrak p$,   is a trivial element  0 or 1.
 
 \medskip
 For otherwise (absurdum hypothesis), say without
 loss of generality
     \begin{equation}
 \label{equation:nine=AH2}
 X_1^\mathfrak A/\mathfrak p \notin \{0,1\},
 \mbox{ whence also }
 ( X_1^{*})^{\mathfrak A}/\mathfrak p \mbox{ does not belong to } \{0,1\}.
 \end{equation}
 By \eqref{equation:eight},
 \begin{eqnarray*}
 \{0,1\}&\ni&\frac{\left(\phi\wedge \bigvee_{i=1}^n
 (X_i\wedge X_i^*)   \right)^\mathfrak A}{\mathfrak p}\\
 {}&=& \frac{\phi^\mathfrak A\wedge \left( \bigvee_{i=1}^n
 (X_i\wedge X_i^*)\right)^\mathfrak A }{\mathfrak p}\\
  {}&=&\frac{\phi^\mathfrak A}{\mathfrak p} \wedge \bigvee_{i=1}^n
  \frac{(X_i\wedge X_i^*)^\mathfrak A}{\mathfrak p}.
 \end{eqnarray*}
 In the MV-chain $A/\mathfrak p$,  from
  \eqref{equation:nine=AH2} we get 
 $$
 0<\frac{(X_1\wedge X_1^*)^\mathfrak A}{\mathfrak p}<1 \mbox{ and 
 for each }j=2,\dots,n, \,\,
 \frac{(X_j\wedge X_j^*)^\mathfrak A}{\mathfrak p}<1.
 $$
 As a consequence,
 $
 0<\left(\bigvee_{i=1}^n \frac{(X_i\wedge X_i^*)^\mathfrak A}{\mathfrak p}\right)<1.
 $
From \eqref{equation:six}, 
we obtain
%
 $0<\frac{\phi^\mathfrak A}{\mathfrak p}\wedge \bigvee_{i=1}^n
  \frac{(X_i\wedge X_i^*)^\mathfrak A}{\mathfrak p}<1,$
 against \eqref{equation:five=AH1}.
 
 \medskip
 Having thus settled our claim, for each $i=1,\dots,n$
we have  $X_i^\mathfrak A/\mathfrak p\in \{0,1\}$ in the MV-chain
$A/\mathfrak p.$
Let $\mathfrak n$ be {\it the only} maximal ideal of 
$A$ above the prime ideal
 $\mathfrak p$,  
as given by  \cite[1.2.12]{cigdotmun}.
Then a fortiori,
$
\left\{\frac{X_i^\mathfrak A}{\mathfrak n},\frac{X_i^{*\mathfrak A}}{\mathfrak n}\right\}
=\{0,1\}\,\,\,\mbox{ in } \,\,\,\frac{A}{\mathfrak n}\,\,\,\,
\mbox{   for each }i=1,\dots,n.
$
Since the set of  elements $\{X_1^\mathfrak A/\mathfrak n,\dots,
X_n^\mathfrak A/\mathfrak n\}$ generates the simple MV-chain
$A/\mathfrak n$, we obtain
$
A/\mathfrak n=\{0,1\}=\mbox{the set of  trivial elements
of $A/\mathfrak n$.}
$
In correspondence with the maximal ideal $\mathfrak n$
of $A$, 
Theorem \ref{theorem:spectral}(v)   yields  
 a maximal ideal $\mathfrak N$ of $\mathfrak A$ such that
 $\mathfrak A/\mathfrak N\cong \mathbb C$.  This 
 contradicts our standing assumption
  \eqref{equation:first-assumption}, thus proving 
  \eqref{equation:three}.  The map 
  $\rho$ defined  in \eqref{equation:rho}
  is the desired polytime reduction of problem $\mathsf{P}_4$ to
  $\mathsf{P}_5$ for $\mathfrak A$ whenever $\mathfrak A$
  has no quotient isomorphic to $\mathbb C$.

\medskip
 
To conclude the proof, there remains to consider the case when
$\mathfrak A$ is simple.

\smallskip
\noindent
  If 
  $\mathfrak A\cong \mathbb C$ then
both    problems  $\mathsf{P}_4$ and $\mathsf{P}_5$
 are polytime, whence they are (trivially) polytime   
 reducible each to the other. 
On the other hand, if 
$\mathfrak A\ncong \mathbb C$,
 the assumed simplicity of
$\mathfrak A$ ensures that 
$\mathfrak A$ has no quotient isomorphic to $\mathbb C$.
The desired polytime reduction of problem $\mathsf{P}_4$ to
  $\mathsf{P}_5$ for $\mathfrak A$ is then obtained arguing
  as in the first part of the  proof of (ii). 
 \end{proof}

 \begin{corollary}
\label{corollary:emme}
Let $\mathfrak M$ be the universal AF- algebra of
\cite[$\S 8$]{mun-jfa}.

(i)  For each $i=1,\dots,5$, problem $\mathsf{P}_i$
  for $\mathfrak M$ 
is   coNP-com\-ple\-te.

\smallskip
(ii)  The nontriviality problem $\mathsf{P}_6$  
 is NP-complete.

\smallskip
(iii)  If   
there is a polytime reduction of    
$\mathsf{P}_i$  to  
$\mathsf{P}_6$  
then   NP\,\,{\rm=}\,\,coNP. 

\smallskip
(iv) The central nontrivial projection problem $\mathsf{P}_7$
  is (trivially) polytime.

\end{corollary}

 \begin{corollary}
\label{corollary:goedel} 
Let $\mathfrak M_1$ be the ``Farey'' AF$\ell$-algebra of
\cite{mun-adv, boc, mun-lincei}.

\smallskip
Problems $\mathsf{P}_1,\dots, \mathsf{P}_7$
for $\mathfrak M_1$
 are decidable in polytime.
\end{corollary}

 \begin{corollary}
 \label{corollary:es}
Problems $\mathsf{P}_1,\dots, \mathsf{P}_7$
are decidable in polynomial time
 for the following AF$\ell$-algebras:

\smallskip
{\rm(i)}   
The Effros-Shen algebra  $\mathfrak F_{\theta}$ for 
 $\theta$   a quadratic
irrational,  or $\theta = 1/{\rm e}$,  or 
$\theta\in \interval\setminus \mathbb Q$
  a real algebraic integer.

\smallskip
{\rm(ii)}
The Effros-Shen algebra  $\mathfrak F_{\theta}$ for
any irrational   $\theta\in [0,1]$ having the following property:
 There is a real  $\kappa>0$ 
 such that for  every $n=0,1,\dots$
the   sequence $[a_0,\dots,a_n]$ of partial quotients
 of  $\theta$ is computable (as a finite list of binary integers) 
 in less than $2^{\kappa n}$  steps.
 \end{corollary}
 
 \begin{corollary}
\label{corollary:bl}
Problems $\mathsf{P}_1,\dots, \mathsf{P}_7$
are decidable in polynomial time
 for
 all   Behncke-Leptin \cite{behlep} algebras $\mathcal A_{m,n}$.
 \end{corollary}

 \bibliographystyle{plain}

\end{document}